\documentclass[11pt,onecolumn,draftcls]{IEEEtran}
\usepackage{graphicx}
\usepackage{epsfig}
\usepackage{subfigure}
\usepackage{amssymb}
\usepackage{amsbsy}
\usepackage{amsmath}
\usepackage{cite}
\usepackage{stfloats}
\usepackage{color}
\usepackage{placeins}
\usepackage{float}
\usepackage{tabularx,colortbl}
\usepackage{times,amsmath,epsfig}
\usepackage{xspace,latexsym,syntonly}
\usepackage{amssymb}
\usepackage{amsfonts}
\usepackage{textcomp}
\usepackage{subfigure}

\newcommand{\E}{\ensuremath{\hbox{\textbf{E}}}}

\newtheorem{theorem}{Theorem}
\newtheorem{lemma}{Lemma}

\newtheorem{definition}{Definition}

\newcommand{\beq}{\begin{equation}}
\newcommand{\eeq}{\end{equation}}
\newcommand{\bea}{\begin{array}}
\newcommand{\ena}{\end{array}}
\newcommand{\bds}{\begin {itemize}}
\newcommand{\eds}{\end {itemize}}
\newcommand{\bdf}{\begin{definition}}
\newcommand{\blm}{\begin{lemma}}
\newcommand{\edf}{\end{definition}}
\newcommand{\elm}{\end{lemma}}
\newcommand{\bthm}{\begin{theorem}}
\newcommand{\ethm}{\end{theorem}}
\newcommand{\bprp}{\begin{prop}}
\newcommand{\eprp}{\end{prop}}
\newcommand{\bcl}{\begin{claim}}
\newcommand{\ecl}{\end{claim}}
\newcommand{\bcr}{\begin{coro}}
\newcommand{\ecr}{\end{coro}}
\newcommand{\bquest}{\begin{question}}
\newcommand{\equest}{\end{question}}

\newcommand{\larrow}{{\larrow}}

\newcommand{\nin}{{\not \in}}

\def\urltilda{\kern -.15em\lower .7ex\hbox{\~{}}\kern .04em}

\begin{document}\title{Asymptotically Optimal Anomaly Detection via\\ Sequential Testing}
\author{Kobi Cohen and Qing Zhao
\thanks{The authors are with the Department of Electrical and Computer Engineering, University of California, Davis. Email: $\left\{\mbox{yscohen, qzhao}\right\}$@ucdavis.edu.}
\thanks{This work was supported in part by National Science Foundation under Grants CCF-1320065 and CNS-1321115.}
\thanks{Part of this work will be presented at the 52nd Annual Allerton Conference on Communication, Control, and Computing, 2014.}
}
\date{}
\maketitle
%
\begin{abstract}
\label{sec:abstract}
Sequential detection of independent anomalous processes among $K$ processes is considered. At each time, only $M$ ($1\leq M\leq K$) processes can be observed, and the observations from each chosen process follow two different distributions, depending on whether the process is normal or abnormal.
Each anomalous process incurs a cost per unit time until its anomaly is identified and fixed.
Switching across processes and state declarations are allowed at all times, while decisions are based on all past observations and actions.
The objective is a sequential search strategy that minimizes the total expected cost incurred by all the processes during the detection process under reliability constraints.
Low-complexity algorithms are established to achieve asymptotically optimal performance as the error constraints approach zero. Simulation results demonstrate strong performance in the finite regime.
\end{abstract}
%
\def\keywords{\vspace{.5em}
{\bfseries\textit{Index Terms}---\,\relax%
}}
\def\endkeywords{\par}
\keywords
Anomaly detection, sequential hypothesis testing, Sequential Probability Ratio Test (SPRT), asymptotic optimality.
\section{Introduction}
\label{sec:intro}

Consider a system consisting of $K$ processes, which can be components (such as routers and paths) in a cyber system, channels in a communication network, potential locations of targets, and sensors monitoring certain events. The state of each process is either normal or abnormal (e.g., the busy/idle state of a channel, the presence or absence of a target or event). Process $k$ is abnormal with prior probability $\pi_k$, independent of other processes. Each abnormal process incurs a cost $c_k$ per unit time until its anomaly is identified and fixed. Normal processes incur no cost. Due to resource constraints, only $M$ ($1\leq M \leq K$) processes can be probed at a time, and the observations from a probed process follow distributions $f_k^{(0)}$ or $f_k^{(1)}$ depending on whether the process is normal or abnormal. The objective is a sequential search strategy that dynamically determines which processes to probe at each time and when to terminate the search so that the total expected cost incurred to the system during the entire detection process is minimized under reliability constraints.

The problem under study finds applications in intrusion detection in cyber systems, spectrum scanning in cognitive radio networks (for quickly catching and utilizing idle channels), target search, and event detection in sensor networks.

\subsection{Main Results}
\label{sec:main_results}

Since observations are drawn in a one-at-a-time manner, the above anomaly detection problem has a clear connection with the classic sequential hypothesis testing problem pioneered by Wald in \cite{Wald_1947_Sequential}. The presence of multiple processes and the objective of minimizing the total cost (rather than the detection delay), however, give the problem another dimension. In addition to quickly declare the state of a process by fully utilizing past observations, the probing order is crucial in minimizing the total cost. It is intuitive that processes with a higher probability of being abnormal and a higher abnormal cost should be probed first. At the same time, it may be desirable to probe processes that require more samples to detect their states (determined by the Kullback-Leibler divergence between $f_k^{(0)}$ and $f_k^{(1)}$) toward the end of the detection process to avoid long delays in catching other potentially abnormal processes.

This anomaly detection problem was first formulated and studied in our prior work \cite{Cohen_2013_Optimal_GlobalSIP, Cohen_2014_Optimal} under the restriction that each process must be probed continuously until its state is declared. In other words, switching across processes is allowed only when the state of the currently probed process is declared. It was shown in \cite{Cohen_2014_Optimal} that the optimal probing strategy is an open-loop strategy that probes processes in a decreasing order of $\frac{\pi_k c_k}{\mathbf{E}(N_k)}$ (referred to as the OL-$\pi c N$ rule), where $\mathbf{E}(N_k)$ is the expected detection time for process $k$. With the restriction that the test of the currently chosen process has to be completed before testing other processes, it is perhaps not surprising that the optimal probing strategy is open-loop: the probing order is predetermined based on prior information $\{\pi_k, c_k, f_k^{(0)} ,f_k^{(1)}\}$, and $K$ uninterrupted sequential tests are carried out, one over each process.

In this paper we relax the restriction on switching across processes during the detection process. We are thus facing a full-blown dynamic problem where at any given time, the decision maker can choose any process whose state has not been declared and the optimal strategy hinges on fully utilizing the entire decision and observation history. In this case, the priority of each process in probing needs to be dynamically updated based on each newly obtained observation. In particular, the probability of each process being abnormal, a key factor in determining the probing order as shown in our prior work \cite{Cohen_2014_Optimal}, should be updated from the prior probability $\pi_k$ to the \emph{a posteriori} probability $\pi_k(n)$ at time $n$ based on all past observations from this process. Consequently, the expected detection time of process $k$ will also dynamically change based on the \emph{a posteriori} probability of being abnormal (see (\ref{eq:sample_size_approx})). Built upon the insights obtained in our prior work \cite{Cohen_2014_Optimal}, we thus propose the following closed-loop $\pi c N$ rule (referred to as CL-$\pi c N$). At each given time $n$, each process is associated with an index $\gamma_k(n)\triangleq\frac{\pi_k(n)c_k}{\mathbf{E}^{(n)}(N_k)}$, where $\pi_k(n)$ is the \emph{a posteriori} probability of process $k$ being abnormal (i.e., the belief) and $\mathbf{E}^{(n)}(N_k)$ is the expected detection time of process $k$ based on $\pi_k(n)$. At each time (except a sparse subsequence of time instants as detailed below), the process with the largest index is probed, and its state is detected via a sequential test using all past observations. The index of this process is also updated (based on the newly obtained observation) for comparison with other processes at the next time instant. To ensure that all processes are sufficiently probed so that the belief $\pi_k(n)$ (consequently the index $\gamma_k(n)$) is a sufficiently accurate indication of the process state, processes are probed in a round-robin fashion at a subsequence of time instants that grows exponentially sparse with time. In other words, a logarithmic order of time is used to explore the state of all processes to ensure the accuracy of the indices $\gamma_k(n)$ used in the remaining majority of time instants. The main technical result of this paper is the establishment of the asymptotic optimality of the CL-$\pi c N$ strategy for $M=1$ for both known and unknown observation models (i.e., whether $\{f_k^{(0)} ,f_k^{(1)}\}$ are known or has unknown parameters). When $M>1$, we show that CL-$\pi c N$ preserves its asymptotic optimality if processes incur the same cost when abnormal (i.e., $c_1=c_2=\cdots=c_K$).
It should be noted that the techniques used in proving the asymptotic optimality under the full-blown dynamic problem considered in this paper are fundamentally different from those used in \cite{Cohen_2014_Optimal} under the switching constraint. The proof for the optimality of the OL-$\pi c N$ policy under the restrictive model in \cite{Cohen_2014_Optimal} is mainly based on an interchange argument, which no longer holds in this fully dynamic problem. In proving the asymptotic optimality of the CL-$\pi c N$ rule under the general model, the key is to show that the average time spent on probing undesired processes (i.e., when noisy observations lead to an inaccurate indication of the process states) does not affect the asymptotic detection time. This is done in two steps. First, we establish the asymptotic lower bound on the total cost that can be achieved by any policy. Second, by upper bounding the tail of the distribution of some ancillary random times, we show that CL-$\pi c N$ achieves the lower bound in the asymptotic regime.

\subsection{Related Work}
\label{ssec:related}

Sequential hypothesis testing was pioneered by Wald in \cite{Wald_1947_Sequential} where he established the Sequential Probability Ratio Test (SPRT) for binary hypothesis testing. For simple hypothesis testing where the observation distributions are known, SPRT is optimal in terms of minimizing the expected sample size under given type $I$ and type $II$ error probability constraints.
Various extensions to M-ary hypothesis testing and testing composite hypotheses have been studied in \cite{Schwarz_1962_Asymptotic, Lai_1988_Nearly, Pavlov_1990_Sequential, Tartakovsky_2002_Efficient, Draglin_1999_Multihypothesis} for a single process. In these cases, asymptotically optimal performance can be obtained in terms of minimizing the expected sample size as the error probability approaches zero.

There are a number of recent studies on sequential detection involving multiple independent processes for various applications (see, for example, \cite{Zhao_2010_Quickest, Li_2009_Restless, Caromi_2013_Fast, Lai_2011_Quickest, Malloy_2012_Sequential, Malloy_2012_Quickest, Tajer_2013_Quick, Geng_2014_Quickest} and references therein). Differing from this work (and our prior work \cite{Cohen_2013_Optimal_GlobalSIP, Cohen_2014_Optimal}), these studies focus on minimizing the total detection delay, which does not translate to minimizing the total system-wide cost in the anomaly detection problem at hand.
The anomaly detection problem also shares similarities with the optimal search and target whereabouts problems as studied in \cite{Zigangirov_1966_Problem, Klimko_1975_Optimal, Dragalin_1996_Simple, Stone_1971_Optimal} under a sequential setting and in \cite{Tognetti_1968_An, Kadane_1971_Optimal, Zhai_2013_Dynamic, Castanon_1995_Optimal} under a fixed sample size setting. The design objectives in these studies again differ from that in this paper. The problem of universal outlier hypothesis testing involving a vector of observations containing coordinates with an outlier distribution was studied in \cite{Li_2013_Universal}.

The anomaly detection problem studied in this paper can be considered as a variation of the sequential design of experiments problem first studied by Chernoff \cite{Chernoff_1959_Sequential}. In this problem, a decision maker aims to infer the state of an underlying phenomenon by sequentially choosing the experiment (thus the observation model) to be conducted at each time among a set of available experiments. Classic and more recent studies of this problem  can be found in
\cite{Bessler_1960_Theory, Nitinawarat_2012_Controlled, Nitinawarat_2013_Controlled, Naghshvar_2013_Active, Naghshvar_2013_Sequentiality, Cohen_2014_Active, Cohen_2014_Quickest}. However, the objective of minimizing the total detection delay makes the problems considered in \cite{Chernoff_1959_Sequential, Bessler_1960_Theory, Nitinawarat_2012_Controlled, Nitinawarat_2013_Controlled, Naghshvar_2013_Active, Naghshvar_2013_Sequentiality, Cohen_2014_Active, Cohen_2014_Quickest} fundamentally different from the one considered in this paper.

\section{System Model and Problem Formulation}
\label{sec:network}

Consider a system consisting of $K$ processes, where each process may be in a normal state (denoted by $H_0$) or abnormal state (denoted by $H_1$). Each process $k$ is abnormal with a prior probability $\pi_k$, independent of other processes. Each abnormal process $k$ incurs a cost $c_k$ ($0\leq c_k<\infty$) per unit time until it is tested and identified. Processes in a normal state do not incur cost. At each given time, only $M$ processes can be probed. We first consider $M=1$. An extension to $M\geq 1$ is discussed in Section \ref{sec:multi}.

When process $k$ is probed at time $n$, a measurement $y_k(n)$ is drawn independently in a one-at-a-time manner. If process $k$ is in a normal state, $y_k(n)$ follows density $f_k^{(0)}$; if process $k$ is abnormal, $y_k(n)$ follows density $f_k^{(1)}$. In section \ref{sec:simple}, we examine the case where the densities $f_k^{(0)}$, $f_k^{(1)}$ are known. In Section \ref{sec:composite} we extend our results to the case where the densities have unknown parameters.

Let $\phi(n)\in\left\{1, 2, ..., K\right\}$ be a selection rule, indicating which process is chosen to be tested at time $n$. Let $\mathbf{y}(n)=\left\{\phi(t), y_{\phi(t)}(t)\right\}_{t=1}^n$ be the set of all the observations and actions up to time $n$. The selection rule $\phi(n)$ is a mapping from $\mathbf{y}(n-1)$ to $\left\{1, 2, ..., K\right\}$. The vector of selection rules over the time series is denoted by $\boldsymbol\phi=(\phi(1), \phi(2), ...)$. Let $\mathbf{1}_k(n)$ be the probing indicator function, where $\mathbf{1}_k(n)=1$ if process $k$ is probed at time $n$ and $\mathbf{1}_k(n)=0$ otherwise.

Let $\tau_k$ be a stopping time (or a stopping rule), which is the time (counted from the beginning of the entire detection process) when the decision maker stops taking observations from process $k$ and declares its state. The vector of stopping times for the $K$ processes is denoted by $\boldsymbol\tau=(\tau_1, ..., \tau_K)$. The random sample size required to make a decision regarding the state of process $k$ is denoted by $N_k$. Let $\delta_k\in\left\{0, 1\right\}$ be a decision rule, indicating the state declaration of process $k$ at time $\tau_k$. $\delta_k=0$ if the decision maker declares that process $k$ is in a normal state, and $\delta_k=1$ if the decision maker declares that process $k$ is in an abnormal state. The vector of decision rules for the $K$ processes is denoted by $\boldsymbol\delta=(\delta_1, ..., \delta_K)$. \vspace{0.2cm}

\begin{definition}
An admissible strategy $\mathbf{s}$ for the sequential anomaly detection problem is given by the
tuple $\mathbf{s}=\left(\boldsymbol{\tau}, \boldsymbol{\delta}, \boldsymbol{\phi}\right)$. \vspace{0.1cm}
\end{definition}

Let \vspace{0.2cm}
\begin{center}
$\bea{l}
\mathcal{H}_0\triangleq\left\{k: 1\leq k \leq K \; , \; \mbox{process $k$ is normal} \right\} \;, \vspace{0.3cm} \\
\mathcal{H}_1\triangleq\left\{k: 1\leq k \leq K \; , \; \mbox{process $k$ is abnormal} \right\} \;,
\ena$ \vspace{0.2cm}
\end{center}
be the sets of the normal and abnormal processes. The objective is to find a strategy $\mathbf{s}$ that minimizes the total expected cost incurred by all the abnormal processes subject to type $I$ (false-alarm) and type $II$ (miss-detection) error constraints for each process: \vspace{0.1cm}
\beq\label{eq:opt1}
\bea{lll}
\displaystyle\inf_{\mathbf{s}} & \E\left\{\displaystyle\sum_{k\in \mathcal{H}_1}{c_k\tau_k}\right\} & \vspace{0.4cm} \\
s.t.           &  P_k^{FA}\leq \alpha_k     & \forall k=1, ..., K, \vspace{0.3cm}\\
               &  P_k^{MD}\leq \beta_k      & \forall k=1, ..., K  \;, \vspace{0.2cm}
\ena
\eeq
where $P_k^{FA}, P_k^{MD}$ denote the false-alarm and miss-detect error probabilities for process $k$, respectively. We point out that the total expected cost defined in (\ref{eq:opt1}) does not include the cost incurred by miss-detected abnormal processes. Since the error constraints are typically required to be small, (\ref{eq:opt1}) well approximates the actual loss in practice.

\section{Anomaly Detection Under Known Observation Models}
\label{sec:simple}

In this section we derive an asymptotically optimal solution for the anomaly detection problem (\ref{eq:opt1}) under the case where the densities $f_k^{(0)}$, $f_k^{(1)}$ are known for all $k$. The proposed probing strategy has a simple closed-loop index form. The index of the currently probed process is updated based on the newly obtained measurement, and the process with the highest index is selected at each given time except a subsequence of time instants that grows exponentially sparse with time. In Section \ref{ssec:computing} we discuss the computation of the index in detail.

\subsection{The CL-$\pi c N$ policy:}
\label{ssec:policy}

In this section we present the CL-$\pi c N$ policy. Let
\beq
\label{eq:sum_LLR}
\displaystyle \ell_k(n)\triangleq\log \frac{f_k^{(1)}(y_k(n))}{f_k^{(0)}(y_k(n))} \;,
\eeq
and
\beq
\label{eq:sum_LLR}
\displaystyle S_k(n)\triangleq\sum_{t=1}^{n}{\ell_k(t)\mathbf{1}_k(t)}
\eeq
be the log-likelihood ratio (LLR) and the observed sum LLRs at time $n$ of process $k$, respectively. Let $\mathcal{K}(n)$ be the set of processes whose states have not been declared up to time $n$. Let $\pi_k(n)$ denote the posterior probability of process $k$ being abnormal at time $n$ (see (\ref{eq:beleif_ind}) for the update of the belief based on a newly obtained measurement). Let $\mathbf{E}^{(n)}(N_k)$ be the expected detection time for process $k$ at time $n$ which dynamically changes due to the changes in the belief $\pi_k(n)$ (see (\ref{eq:sample_size_approx})). Define
\beq
\gamma_k(n)\triangleq
\begin{cases}
\displaystyle\frac{\pi_k(n)c_k}{\mathbf{E}^{(n)}(N_k)} \;, \; \mbox{if} \; k\in\mathcal{K}(n) \;, \\
\displaystyle\hspace{0.6cm} 0 \hspace{0.8cm} \;, \;\mbox{otherwise} \;.
\end{cases}
\eeq
Let $\mathcal{N}_s=\left\{n_1, n_2, ...\right\}$ be a set of time instants that grows exponentially sparse with time (i.e., the cardinality of $\mathcal{N}_s$ grows at a logarithmic rate with time). The CL-$\pi c N$ policy selects the process with the highest index $\gamma_k(n)$ at all times except at time instants in $\mathcal{N}_s$. During the subsequence $\mathcal{N}_s$, all processes whose states have not been declared are probed in a round robin fashion. Specifically,
\beq
\label{eq:selection_simple}
\displaystyle\phi(n)=
\begin{cases}
\displaystyle\hspace{0.1cm}\arg\max_{k}\;\gamma_k(n) \;,\; \mbox{if\;} n\nin\mathcal{N}_s,
\vspace{0.2cm}\\\hspace{0.8cm}
\displaystyle \;r(n)\; \hspace{0.8cm},\; \mbox{if\;} n=n_i \;\forall\; i=2, 3, ... \;.
\end{cases}
\eeq
The function $r(n)$ is given by:
\beq
\bea{l}
\label{eq:round_robin}
\displaystyle r(n)
=\left[\left(\phi(n_{i-1})+u(n)\right)
\mbox{\;mod\;} K\right]+1 \;,
\ena
\eeq
where $u(n)=\min\left(0, 1, ..., K-1\right)$ s.t $r(n)\in\mathcal{K}(n)$, mod denotes the modulo operator, and $\phi(n_{1})=1$. Note that processes are no longer probed once their state has been declared. The round-robin probing subsequence $\mathcal{N}_s$ is to ensure all processes are sufficiently explored. We set\footnote{Note that duplicate values in $\mathcal{N}_s$ are removed.} $\mathcal{N}_s=\left\{\lceil \zeta^\ell\rceil\right\}_{\ell=1}^{\infty}$, where $\zeta>1$ is a design parameter (for details see Section \ref{ssec:computing}). We point out that this idea of introducing an exploration subsequence to ensure sufficient learning has also been used in \cite{Nitinawarat_2013_Controlled, Vakili_2013_Deterministic}.

Following the Wald's SPRT \cite{Wald_1947_Sequential}, $S_{\phi(n)}(n)$ is compared to boundary values $A_{\phi(n)}, B_{\phi(n)}$ as follows: \vspace{0.3cm}
\begin{itemize}
  \item If $S_{\phi(n)}(n)\in\left(A_{\phi(n)}, B_{\phi(n)}\right)$, then $\phi(n)\in\mathcal{K}(n+1)$ (i.e., continue to take observations from process $\phi(n)$ according to the selection rule (\ref{eq:selection_simple}) at time $n+1$). \vspace{0.3cm}
  \item If $S_{\phi(n)}(n)\geq B_k$, stop taking observations from process $k$ and declare it as abnormal (i.e., $\tau_{\phi(n)}=n$, $\delta_{\phi(n)}=1$ and $\phi(n)\nin\mathcal{K}(n')$ for all $n'>n$). \vspace{0.3cm}
  \item If $S_{\phi(n)}(n)\leq A_k$, stop taking observations from process $k$ and declare it as normal (i.e., $\tau_{\phi(n)}=n$, $\delta_{\phi(n)}=0$ and $\phi(n)\nin\mathcal{K}(n')$ for all $n'>n$). \vspace{0.3cm}
\end{itemize}

The boundary values $A_k$ and $B_k$ are determined such that the error constraints are satisfied. In general, the exact computation of the boundary values is very laborious under the finite regime. Nevertheless, Wald's approximation can be applied to simplify the computation \cite{Wald_1947_Sequential}:
\beq
\label{eq:boundary_approx}
\bea{l}
A_k\approx \displaystyle\log\left(\frac{\beta_k}{1-\alpha_k}\right)\;, \vspace{0.3cm} \\
B_k\approx \displaystyle\log\left(\frac{1-\beta_k}{\alpha_k}\right)\;.
\ena
\eeq
Wald's approximation performs well for small $\alpha_k, \beta_k$ and is asymptotically optimal as the error probability approaches zero. Since type $I$ and type $II$ errors are typically required to be small, Wald's approximation is widely used in practice \cite{Wald_1947_Sequential}.

Note that CL-$\pi c N$ is a closed-loop strategy, where the index $\gamma_k(n)$ is updated at each given time based on past observations and actions and the next process is selected accordingly. It can be seen that CL-$\pi c N$ handles the well-known trade-off between exploration and exploitation. The decision maker spends a logarithmic order of time by selecting the processes in a round-robing manner to explore their states and guard against miss-detected abnormal processes. On the other hand, at times $n\nin\mathcal{N}_s$, it exploits the information gathered so far to select the process according to the updated index $\gamma_k(n)$ at time $n$. The index form under the CL-$\pi c N$ policy which dynamically updates the priority of the processes is intuitively satisfying. We should prioritize processes that incur higher costs to the system when abnormal. Furthermore, the priority of a process should be increased as the updated belief of it being abnormal increases during the detection process. It is also desirable to place processes that require longer testing time toward the end of the testing process since their detection time contributes to the cost of every abnormal process that has not been identified. Thus, the priority of a process increases as the updated expected detection time decreases. Note that the sequential test uses an SPRT-based method with memory to minimize the expected sample size for every process. When switching back to a previously visited process (say $k$) at time $n$, the sequential test uses the sum LLRs $S_k(n)$ in decision making to exploit all past observations obtained during previous visits.

\subsection{Performance Analysis}
\label{ssec:performance}

In this section we analyze the performance of the CL-$\pi c N$ policy. Let
\beq
P_e^{max}\triangleq\max\left(\alpha_1, \beta_1, ..., \alpha_K, \beta_K\right) \;.
\eeq
The following theorem shows that CL-$\pi c N$ is asymptotically optimal in terms of minimizing the expected cost as the error probability approaches zero. When deriving asymptotic we assume regularity conditions on the error constraints, as discussed in App. \ref{app}.
 \vspace{0.2cm}
\begin{theorem}
\label{th:asymptotic_optimality_simple}
Let $\E(C^*), \E(C(\mathbf{s}))$ be the expected costs under CL-$\pi c N$ and any other policy $\mathbf{s}$, respectively. Then\footnote{The notation $g\sim f$ as $P_e^{max}\rightarrow 0$ implies $\displaystyle\lim_{P_e^{max}\rightarrow 0}g/f=1$},
\beq
\bea{l}
\displaystyle \E(C^*)\;\sim\;\inf_{\mathbf{s}}\;\E(C(\mathbf{s})) \;\;\mbox{as}\;\; P_e^{max}\rightarrow 0 \;. \vspace{0.2cm}
\ena
\eeq
\end{theorem}
\begin{proof}
See Appendix \ref{app:asymptotic_optimality_simple}. \vspace{0.2cm}
\vspace{0.2cm}
\end{proof}

\subsection{Implementation}
\label{ssec:computing}

In this section we discuss the implementation of the proposed policy. At each time $n$, the decision maker updates the indices and the sum LLRs for the currently probed processes, and also sorts the indices for selecting the next process. Sorting the indices can be done by $O(K \log K)$ time via a sorting algorithm. Updating the indices and the sum LLRs (for the general case where $M$ processes are probed at a time) requires $O(M)$ time.

We now consider the computation of the index $\gamma_k(n)=\pi_k(n)c_k/\mathbf{E}^{(n)}(N_k)$. The posterior probability of process $k$ being abnormal can be updated at time $n+1$ based on the Bayes rule:
\beq
\label{eq:beleif_ind}
\bea{l}
\pi_k(n+1)=\displaystyle\left(1-\mathbf{1}_k(n)\right)\pi_k(n) \vspace{0.2cm}\\ \hspace{0.5cm}
+\displaystyle\frac{\mathbf{1}_k(n)\pi_k(n)f_k^{(1)}(y_k(n))}
                {\pi_k(n)f_k^{(1)}(y_k(n))+\left(1-\pi_k(n)\right)f_k^{(0)}(y_k(n))}\;.
\ena
\eeq
Note that at time $n+1$, only the index of the process that was probed at time $n$ needs to be updated. The expected sample size $\mathbf{E}^{(n)}(N_k)$ at time $n$ depends on the currently belief value:
\beq
\label{eq:sample_size_approx}
\bea{l}
\mathbf{E}^{(n)}(N_k)
= \pi_k(n)\mathbf{E}(N_k|H_1)+(1-\pi_k(n))\mathbf{E}(N_k|H_0) \;,
\ena
\eeq
where $\mathbf{E}(N_k|H_i)$ is the expected detection time for process $k$ conditioned on its state $H_i$. In general, it is difficult to obtain a closed-form expression for $\mathbf{E}^{(n)}(N_k|H_i)$ under the finite regime. However, Wald's approximation can be applied to simplify the computation \cite{Wald_1947_Sequential}:
\beq
\label{eq:sample_size_approx_H}
\bea{l}
\mathbf{E}(N_k|H_0)\approx \displaystyle\frac{\left(1-\alpha_k\right)\log\frac{1-\alpha_k}{\beta_k} - \alpha_k \log\frac{1-\beta_k}{\alpha_k}}{D(f_k^{(0)}||f_k^{(1)})} \;, \vspace{0.2cm} \\
\mathbf{E}(N_k|H_1)\approx \displaystyle\frac{\left(1-\beta_k\right)\log\frac{1-\beta_k}{\alpha_k} - \beta_k \log\frac{1-\alpha_k}{\beta_k}}{D(f_k^{(1)}||f_k^{(0)})} \;,
\ena
\eeq
where $D(f_k^{(i)}||f_k^{(j)})=\mathbf{E}_i\left(\log\frac{f_k^{(i)}(y_k(1))}{f_k^{(j)}(y_k(1))}\right)$ denotes the Kullback-Leibler (KL) divergence between the hypotheses $H_i$ and $H_j$. This approximation approaches the exact expected sample size for small $\alpha_k, \beta_k$. We point out that asymptotic optimality of the probing strategy is preserved as long as the required \emph{order} of the indices is preserved. Therefore, computing the exact expected remaining detection time of a process during a sequential test is not required. Using the Wald's approximation to the entire detection time when computing the indices at each given time is sufficient for obtaining asymptotic optimality.

Next, we discuss the design parameter $\zeta>1$ used in the exploration subsequence $\mathcal{N}_s$. Note that as $\zeta$ approaches $1$, the round-robin selection rule is executed more frequently. It is shown in App. \ref{app} that asymptotic optimality of CL-$\pi c N$ holds when $\zeta$ is set sufficiently close to $1$ to ensure that the round-robin probing gathers sufficient information so that the index $\gamma_k(n)$ is a sufficiently accurate indication of the process state. In the finite regime, however, $\zeta$ must be designed judiciously for better performance. Intuitively speaking, one should increase $\zeta$ as the sample sizes required to declare the process states decrease to reduce the time spent during the round-robin selection rule. For instance, consider the extreme case where only a single observation is required to declare the process states (i.e., the KL divergences between the observation distributions are sufficiently large). Therefore, switching between processes is done only when the state of the currently probed process is declared. In this extreme case, the optimal probing strategy is to test the processes in decreasing order of $\pi_k c_k$. Hence, it is desirable to set $\zeta$ sufficiently high in that case so that only the first line in (\ref{eq:selection_simple}) will be executed to obtain optimal performance.

\section{Anomaly Detection Under Unknown Observation Models}
\label{sec:composite}

In the previous section we focused on the case where the densities under both hypotheses are known. For that case, the sum LLRs was used by every process to design stopping and decision rules based on Wald's SPRT which minimizes the expected sample size for detection. In this section we consider the case where the densities have unknown parameters. While the SPRT applies to the latter case as well with minor modifications, it is highly sub-optimal in general. Therefore, in what follows we focus on asymptotically optimal tests in terms of minimizing the sample size as the error probability approaches zero.

Let $\theta_k$ be an unknown parameter (or a vector of unknown parameters) of process $k$. The observations $\left\{y_k(i)\right\}_{i\geq1}$ are drawn from a common density $f_k\left(y|\theta_k\right)$, $\theta_k\in\Theta_k$, where $\Theta_k$ is the parameter space of process $k$. If process $k$ is in a normal state, then $\theta_k\in\Theta_k^{(0)}$; if process $k$ is in an abnormal state, then $\theta_k\in(\Theta_k\backslash\Theta_k^{(0)})$. Let $\Theta_k^{(0)}$, $\Theta_k^{(1)}$ be disjoint subsets of $\Theta_k$, where $I_k=\Theta_k\backslash(\Theta_k^{(0)}\cup\Theta_k^{(1)})\neq\emptyset$ is an indifference region\footnote{The assumption of an indifference region is widely used in the theory of sequential hypothesis testing to derive asymptotically optimal performance. Nevertheless, in some cases this assumption can be removed. For more details, the reader is referred to \cite{Lai_1988_Nearly}.}.
When $\theta_k\in I_k$, the detector is indifferent regarding the state of process $k$. Hence, there are no constraints on the error probabilities for all $\theta_k\in I_k$. The hypothesis test regarding process $k$ is to test $\theta_k\in\Theta_k^{(0)}$ \; against \; $\theta_k\in\Theta_k^{(1)}$. Reducing $I_k$ increases the sample size.

Asymptotically optimal sequential tests for a single process have been widely studied in the literature, where the key idea is to use the maximum likelihood estimate (MLE) of the unknown parameters to perform a one-sided sequential test to reject $H_0$ and a one-sided sequential test to reject $H_1$. It is assumed that regularity conditions on the distribution hold to guarantee consistency of the MLE \cite{Kay_1993_Fundamentals}. One way to perform the sequential test is to use the Generalized Likelihood Ratio (GLR) statistics. Let $\mathbf{y}_k(n)=(y_k(1), ..., y_k(n))$ be the vector of observations for process $k$ by time $n$. For $i,j \in\left\{0, 1\right\}$ and $i\neq j$, let
\beq
\label{eq:GLR}
S_k^{(i), GLR}(n)=\displaystyle\sum_{r=1}^{n}\log\frac{f_k(y_k(r)|\hat{\theta}_k(n))}
                        {f_k(y_k(r)|\hat{\theta}_k^{(j)}(n))}
\eeq
be the GLR statistics used to declare hypothesis $H_i$ (i.e., reject hypothesis $H_j$) at stage $n$, where $\hat{\theta}_k(n)
=\arg\max_{\theta_k\in\Theta_k}{f_k\left(\mathbf{y}_k(n)|\theta_k\right)}$ and $\hat{\theta}_k^{(j)}(n)=\arg\max_{\theta_k\in\Theta_k^{(j)}}{f_k\left(\mathbf{y}_k(n)|\theta_k\right)}$
are the Maximum-Likelihood (ML) estimates of the parameters over the parameter spaces $\Theta_k$ and $\Theta_k^{(j)}$ at stage $n$, respectively.  \\
Another way is to use the Adaptive Likelihood Ratio (ALR) statistics. For $i,j \in\left\{0, 1\right\}$ and $i\neq j$, let
\beq
\label{eq:ALR}
S_k^{(i), ALR}(n)=\displaystyle\sum_{r=1}^{n}\log\frac{f_k(y_k(r)|\hat{\theta}_k(r-1))}
                         {f_k(y_k(r)|\hat{\theta}_k^{(j)}(n))}
\eeq
be the ALR statistics used to declare hypothesis $H_i$ at stage $n$.
Let $S_k^{(i)}(n)$ be the chosen statistics and let
\beq\label{eq:ALR_stopping}
\bea{l}
N_k^{(i)}=\displaystyle\inf\left\{ \; n \; :  S_k^{(i)}(n) \geq B_k^{(i)}
\right\}
\ena
\eeq
be the stopping rule used to declare hypothesis $H_i$, where $B_k^{(i)}$ is the boundary value.
For each process $k$, the decision maker stops the sampling when $N_k=\min\left\{N_k^{(0)}, N_k^{(1)}\right\}$. If $N_k=N_k^{(0)}$, process $k$ is declared as normal. If $N_k=N_k^{(1)}$, process $k$ is declared as abnormal.
The advantage of using the ALR statistics is that setting $B_k^{(0)}=\log\frac{1}{\alpha_k}$, $B_k^{(1)}=\log\frac{1}{\beta_k}$ satisfies the error probability constraints in (\ref{eq:opt1}). However, such a simple setting cannot be applied when using the GLR statistics. Thus, implementing sequential tests using the ALR statistics is much simpler than using the GLR statistics. The disadvantage of using the ALR statistics is that poor early estimates (from a small number of observations) can
never be revised even after a large number of observations have been collected. For more details on sequential tests involving densities with unknown parameters, the reader is referred to \cite{Schwarz_1962_Asymptotic, Lai_1988_Nearly, Pavlov_1990_Sequential, Tartakovsky_2002_Efficient}.

\subsection{The CL-$\pi c N$ Policy}
\label{ssec:modifying}

With some modifications, the CL-$\pi c N$ policy proposed in Sec. \ref{sec:simple} can be applied to the case with unknown observation models. Let $S_k^{(i)}(n)$ be the GLR (\ref{eq:GLR}) or ALR (\ref{eq:ALR}) statistics used in the test. Define
\beq
\hat{\gamma}_k(n)\triangleq
\begin{cases}
\displaystyle\frac{\hat{\pi}_k(n)c_k}{\hat{\mathbf{E}}^{(n)}(N_k)} \;, \; \mbox{if} \; k\in\mathcal{K}(n) \;, \\
\displaystyle\hspace{0.6cm} 0 \hspace{0.8cm} \;, \;\mbox{otherwise} \;,
\end{cases}
\eeq
where $\hat{\pi}_k(n)$ denotes the estimated posterior probability of process $k$ being abnormal and $\hat{\mathbf{E}}^{(n)}(N_k)$ the updated expected detection time for process $k$ at time $n$ (see Sec. \ref{ssec:computing_composite} for the computation of the index). Similar to (\ref{eq:selection_simple}), the selection rule is given by:
\beq
\label{eq:selection_composite}
\displaystyle\phi(n)=
\begin{cases}
\displaystyle\hspace{0.1cm}\arg\max_{k}\;\hat{\gamma}_k(n) \;,\; \mbox{if\;} n\nin\mathcal{N}_s, \vspace{0.2cm}\\\hspace{0.9cm}
\displaystyle r(n) \hspace{0.9cm},\; \mbox{if\;} n=n_i \;\forall\; i=2, 3, ... \;,
\end{cases}
\eeq
where $r(n)$ is given in (\ref{eq:round_robin}) and $\phi(n_{1})=1$.
Then, $S_{\phi(n)}^{(i)}(n)$ is compared to boundary values $B^{(0)}_{\phi(n)}, B^{(1)}_{\phi(n)}$ as follows: \vspace{0.3cm}
\begin{itemize}
  \item If $S_k^{(0)}(n)<B_k^{(0)}$ and $S_k^{(1)}(n)<B_k^{(1)}$, then $\phi(n)\in\mathcal{K}(n+1)$ (i.e., continue to take observations from process $\phi(n)$ according to the selection rule (\ref{eq:selection_composite}) at time $n+1$).  \vspace{0.3cm}
  \item If $S_k^{(1)}(n)\geq B_k^{(1)}$, stop taking observations from process $k$ and declare it as abnormal (i.e., $\tau_{\phi(n)}=n$, $\delta_{\phi(n)}=1$ and $\phi(n)\nin\mathcal{K}(n')$ for all $n'>n$). \vspace{0.3cm}
  \item If $S_k^{(0)}(n)\geq B_k^{(0)}$, stop taking observations from process $k$ and declare it as normal (i.e., $\tau_{\phi(n)}=n$, $\delta_{\phi(n)}=0$ and $\phi(n)\nin\mathcal{K}(n')$ for all $n'>n$). \vspace{0.2cm}
\end{itemize}

\subsection{Performance Analysis}
\label{ssec:performance_composite}
The following theorem shows that the proposed policy is asymptotically optimal in terms of minimizing the expected cost as the error probability approaches zero. For purposes of analysis we consider the model in \cite{Chernoff_1959_Sequential}, where $\theta_k$ can take only a finite number of values. \vspace{0.2cm}
\begin{theorem}
\label{th:asymptotic_optimality_composite}
Let $\E(C^*), \E(C(\mathbf{s}))$ be the expected costs under CL-$\pi c N$ and any other policy $\mathbf{s}$, respectively. Then, \vspace{0.2cm}
\beq
\bea{l}
\displaystyle \E(C^*)\;\sim\;\inf_{\mathbf{s}}\;\E(C(\mathbf{s})) \;\;\mbox{as}\;\; P_e^{max}\rightarrow 0 \;. \vspace{0.2cm}
\ena
\eeq
\end{theorem}
\begin{proof}
See Appendix \ref{app:asymptotic_optimality_composite}. \vspace{0.2cm}
\vspace{0.2cm}
\end{proof}

\subsection{Implementation}
\label{ssec:computing_composite}

In this section we discuss the implementation of the proposed policy when the densities have unknown parameters. At each time $n$, the decision maker updates the indices and the GLR/ALR statistics for the currently probed processes (i.e., $M$ processes in general), and also sorts the indices for selecting the next process. Sorting the indices can be done by $O(K \log K)$ time via a sorting algorithm. Note that when the densities have unknown parameters, the updated belief must be computed with respect to the current MLE. In cases where the unknown parameters can take a small number $L$ of values, the decision maker can update and store the beliefs for the $L$ values. Thus, $O(LM)$ time is required instead of $O(M)$. However, if the support has infinite values, then the index must be computed at each time $n$ using the past $n$ observations, which generally requires $O(Mn)$ time (unless a quantization on the support is applied).
In general, the estimated belief of process $k$ can be updated at time $n+1$ as follows:
\beq
\label{eq:beleif_ind_composite}
\bea{l}
\hat{\pi}_k(n+1)=\displaystyle\left(1-\mathbf{1}_k(n)\right)\hat{\pi}_k(n) \vspace{0.2cm}\\ \hspace{0.5cm}
+\displaystyle\frac{\mathbf{1}_k(n)\hat{\pi}_k(n)\hat{f}_k^{(1)}(y_k(n))}
                {\hat{\pi}_k(n)\hat{f}_k^{(1)}(y_k(n))
                +\left(1-\hat{\pi}_k(n)\right)\hat{f}_k^{(0)}(y_k(n))}\;,
\ena
\eeq
where $\hat{\pi}_k(1)=\pi_k(1)$ and $\hat{f}_k^{(1)}(y_k(r))\triangleq f_k(y_k(r)|\hat{\theta}_k^{(1)}(n))$, $\hat{f}_k^{(0)}(y_k(r))\triangleq f_k(y_k(r)|\hat{\theta}_k^{(0)}(n))$ for all $1\leq r\leq n$. Note that computing $\hat{\pi}_k(n+1)$ at time $n+1$ requires $n$ computations with the current ML estimate of the parameter.

In general, it is difficult to obtain a closed-form expression for $\hat{\mathbf{E}}^{(n)}(N_k)$ under the finite regime. However, we can use the asymptotic property of the sequential tests to obtain a closed-form approximation to $\hat{\mathbf{E}}^{(n)}(N_k)$ based on the ML estimate of the parameter, which approaches the exact expected sample size as the error probability approaches zero. Let $D_k(\hat{\theta}_k(n)||\theta)\triangleq\mathbf{E}_{\hat{\theta}_k(n)}\left(\log\frac{f_k(y_k(n)|\hat{\theta}_k(n))}{f_k(y_k(n)|\theta)}\right)$ be the KL divergence between $f_k(y_k(n)|\hat{\theta}_k(n))$ and $f_k(y_k(n)|\theta)$, where the expectation is taken with respect to $f_k(y_k(n)|\hat{\theta}_k(n))$ and let $D_k(\hat{\theta}_k(n)||\Theta_k^{(i)})=\inf_{\theta\in\Theta_k^{(i)}}D_k(\hat{\theta}_k(n)||\theta)$. Then, the estimated expected sample size required to make a decision regarding the state of process $k$ is given by:
\beq
\bea{l}
\displaystyle\hat{\mathbf{E}}^{(n)}(N_k) \vspace{0.3cm}
=
\begin{cases}
\displaystyle\frac{B_k^{(0)}}{D_k\left(\hat{\theta}_k(n)||\Theta_k^{(1)}\right)} \;\;, \;\; \mbox{if\;\;} \hat{\theta}_k(n))\in\Theta_k^{(0)} \;, \\
\displaystyle\frac{B_k^{(1)}}{D_k\left(\hat{\theta}_k(n)||\Theta_k^{(0)}\right)} \;\;, \;\; \mbox{if\;\;} \hat{\theta}_k(n))\in\Theta_k^{(1)} \;,
\end{cases}
\ena
\eeq
which is guaranteed to be the asymptotic sample size under various families of distributions with unknown parameters (e.g., exponential, multi-variate distributions and general distributions when the unknown parameters can take a finite number of values) as the error probabilities approach zero~\cite{Lai_1988_Nearly, Pavlov_1990_Sequential, Tartakovsky_2002_Efficient, Chernoff_1959_Sequential, Nitinawarat_2012_Controlled}.\\

It should be noted that implementing the open-loop policy OL-$\pi c N$ \cite{Cohen_2014_Optimal} when the densities have unknown parameters requires a priori knowledge of the parameter's distribution (since the testing order is predetermined and switching between processes is allowed only when the state of the currently probed process is declared). However, under CL-$\pi c N$, the testing order is updated dynamically depending on all past observations and actions. As a result, estimating the detection time at time $n$ does not require a priori knowledge of $\theta_k$ since $\hat{\theta}_k(n)$ converges to its true value.

\section{Extension to Multi-Process Probing}
\label{sec:multi}

In this section we extend the results reported in the previous sections to the case where more than one process can be probed simultaneously (i.e., $M\geq 1$). For the ease of presentation, we will focus on the case where the observation models are known. However, the results apply to the case where the densities have unknown parameters.

Let $\boldsymbol{\sigma}(n)=(\sigma_1(n), ..., \sigma_K(n))$ be a permutation of $\left\{1, ..., K\right\}$ at time $n$ such that:
\beq
\gamma_{\sigma_1(n)}(n)\geq\gamma_{\sigma_2(n)}(n)\geq\cdots\geq\gamma_{\sigma_K(n)}(n)\;.
\eeq
The CL-$\pi c N$ policy selects the processes with the $M$ highest indices at all times except times $\mathcal{N}_s$ at which processes are probed in a round-robin manner, i.e.,
\beq
\label{eq:selection_simple_multiple}
\displaystyle\phi(n)=
\begin{cases}
\displaystyle\hspace{0.1cm} (\sigma_1(n), ..., \sigma_M(n)) \;,\; \mbox{if\;} n\nin\mathcal{N}_s,
\vspace{0.2cm}\\
\displaystyle\hspace{0.1cm} (r_1(n), ..., r_M(n)) \;\;,\; \mbox{if\;} n=n_i \;\forall\; i=2, 3, ... \;.
\end{cases}
\eeq
The functions $(r_1(n), ..., r_M(n))$ select the processes whose states have not been declared by time $n$ in a around-robin manner and are given recursively by:
\beq
\bea{l}
\label{eq:round_robin_multi}
\displaystyle r_1(n)
=\left[\left(r_M(n_{i-1})+u_1(n)\right) \hspace{0.0cm}
\mbox{\;mod\;} K\right]+1\;, \vspace{0.3cm} \\
\displaystyle r_i(n)
=\left[\left(r_{i-1}(n_{i})+u_i(n)\right)
\mbox{\;mod\;} K\right]+1 \;\;\;, \;\; i=2, ..., M \;,
\ena
\eeq
where $u_i(n)=\min\left(0, 1, ..., K-i\right)$ s.t $r_i(n)\in\mathcal{K}(n)$, mod denotes the modulo operator, and $r_i(n_{1})=i$. If there is no solution to $r_i(n)$ (i.e., when $|\mathcal{K}(n)|<M$), then $r_i(n)$ remains empty. Then, sequential tests with memory are executed for the selected processes as described in the previous sections. The following theorem shows that if $c_k=c_{k'}$ holds for all $1\leq k, k'\leq K$, then CL-$\pi c N$ is asymptotically optimal.

\begin{theorem}
\label{th:asymptotic_optimality_multi}
Assume that $c_k=c_{k'}$ holds for all $1\leq k, k'\leq K$. Let $\E(C^*), \E(C(\mathbf{s}))$ be the expected costs under CL-$\pi c N$ and any other policy $\mathbf{s}$, respectively. Then,
\beq
\bea{l}
\displaystyle \E(C^*)\;\sim\;\inf_{\mathbf{s}}\;\E(C(\mathbf{s})) \;\;\mbox{as}\;\; P_e^{max}\rightarrow 0 \;. \vspace{0.2cm}
\ena
\eeq
\end{theorem}
\begin{proof}
See Appendix \ref{app:asymptotic_optimality_multi}. \vspace{0.2cm}
\vspace{0.2cm}
\end{proof}

\section{Numerical Examples}
\label{sec:simulation}

In this section we present numerical examples to illustrate the performance of the proposed CL-$\pi c N$ policy. We test the following hypotheses: under normal state, the observations from process $k$ follow Poisson distribution $y_k(n)\sim\mathrm{Poi}(\theta_k^{(0)})$, where under abnormal state the observations follow Poisson distribution $y_k(n)\sim\mathrm{Poi}(\theta_k^{(1)})$. This model applies to cyber-systems, where the observations from a probed component represent packet arrival rate under normal state or under reduction of quality attacks as in \cite{Onat_2005_Intrusion}. We compare the optimal open-loop probing strategy OL-$\pi c N$ developed in \cite{Cohen_2014_Optimal} with  CL-$\pi c N$. We set the following parameters unless otherwise specified: $c_k=\theta_k^{(0)}$ (i.e., the cost represents the normal expected traffic over the component). Thus, in this setting minimizing the total expected cost minimizes the maximal damage to the network in terms of the expected number of failed packets during a denial of service attack. Only a single component is probed at a time (i.e., $M=1$). The design parameter for the round-robin exploration is set to $\zeta=1.7$. The error constraints are set to $P_k^{FA}=10^{-3}, P_k^{MD}=10^{-6}$ and the a priori probabilities of the components being abnormal are set to $\pi_k=0.5$ for all $k$.

First, we simulate the case where $\theta_k^{(0)}$ are equally spaced in the interval $[10 , 20]$, where
$\theta_k^{(1)}=1.5\cdot\theta_k^{(0)}$ with probability $0.5$ and $\theta_k^{(1)}=1.2\cdot\theta_k^{(0)}$ with probability $0.5$. This models the situation where both strong and weak deviations from the normal state may occur. We implemented CL-$\pi c N$ under densities with unknown parameters (i.e., the level of deviation from the normal state in this scenario) as described in Section \ref{sec:composite}.
The performance of the algorithms is presented in Fig. \ref{fig:fig_composite}. It can be seen that CL-$\pi c N$ saves roughly $40\%$ of the average total cost as compared to OL-$\pi c N$.
Second, we simulate the case where $M=5$ components are probed at a time. We set $\theta_k^{(0)}=10$ for $k=1, 2, ..., K/2$, $\theta_k^{(0)}=20$ for $k=K/2+1, K/2+2, ..., K$ and $\theta_k^{(1)}=1.5\cdot\theta_k^{(0)}$. Note that in that case, asymptotic optimality is an open question due to different costs across the processes. The CL-$\pi c N$ is implemented via multi-process probing as described in Section \ref{sec:multi}. The performance of the algorithms is presented in Fig. \ref{fig:fig3}. It can be seen that CL-$\pi c N$ significantly outperforms OL-$\pi c N$ under this setting as well.

Next, we examine the interesting case where any switching to components $k=1, ..., K/2$ adds a delay $d_1$, while any switching to components $k=K/2+1, ..., K$ adds a delay $d_2$. This models the situation (as in power systems or communication networks for instance) where monitoring different components requires an initialization process which results in different delays. Note that for any fixed delay incurred by switching among components, the CL-$\pi c N$ preserves its optimality in the asymptotic regime. This can be verified by Lemmas \ref{lemma:T1}, \ref{lemma:asymptotic_simple} showing that the time spent until the desired asymptotic order is preserved (where switching no longer occurs) is small enough and does not affect the asymptotic expected cost. In the finite regime, however, one should reduce the number of switchings as the delay incurred in switching increases. As discussed in \cite{Cohen_2014_Optimal}, the advantage of OL-$\pi c N$ is that only $K-1$ switchings among components are required. Hence, we expect OL-$\pi c N$ to outperform CL-$\pi c N$ in the finite regime as the delay incurred in switching increases. We set $\theta_k^{(0)}=10$ for $k=1, 2, ..., K/2$, $\theta_k^{(0)}=20$ for $k=K/2+1, K/2+2, ..., K$ and $\theta_k^{(1)}=1.5\cdot\theta_k^{(0)}$. We set $d_1=1$. Let $\rho=\frac{C_{CL}}{C_{OL}}$, where $C_{CL}$, $C_{OL}$, are the average total costs under CL-$\pi c N$ and OL-$\pi c N$, respectively. The performance of the algorithms is presented in Fig. \ref{fig:fig2}, where $d_2$ ranges between $0$ to $8$ time units. It can be seen that CL-$\pi c N$ saves roughly $30\%-40\%$ of the average total cost as compared to OL-$\pi c N$ when $d_2=0$. On the other hand, OL-$\pi c N$ may be preferred for $d_2>8$.

The next numerical example demonstrates the trade-off curve between the average total cost and the error probabilities (i.e., a Bayes risk) to quantify the threshold effects of the sequential tests. We set $K=10$ and $\theta_k^{(0)}=10$, $\theta_k^{(1)}=15$, $c_k=1$, $\pi_k=0.5$ for all $k$. We assign a cost $c_e$ for a wrong declaration and examine the following normalized (by $c_e$) Bayes risk: $R\triangleq\sum_{k\in\mathcal{H}_1}\left[\frac{1}{c_e}\tau_k+\left(P_k^{FA}+P_k^{MD}\right)\right]$. The log-Bayes risk is presented in Fig. \ref{fig:Bayes_risk} as a function of $\log c_e$, with the corresponding error probabilities $P_e$. As expected, as the cost for a wrong declaration $c_e$ increases, the error probability decreases. Note also that the Bayes risk decreases as $c_e$ increases. Intuitively speaking, this result follows from the fact that the minimal sample size under a sequential testing has the order of $\log(c_e)$, and $P_e$ has the order of $1/c_e$ \cite{{Chernoff_1959_Sequential}}. Thus, the log-Bayes risk decreases approximately linearly with $\log c_e$ as $c_e$ increases.

Finally, we demonstrate the loss of optimality in the asymptotic regime when the round-robin selection rule is not executed. We set $K=2$, $\theta_1^{(0)}=\theta_2^{(0)}=10$, $\theta_1^{(1)}=10.1, \theta_2^{(1)}=10.3$ (i.e., small deviations from normal states are required to be detected), $\pi_1=0.9, \pi_2=0.1$, $c_1=c_2=1$.
We simulated CL-$\pi c N$ under $\zeta=1.005$ (i.e., the round-robin scheduling is executed very frequently) and $\zeta\rightarrow\infty$ (i.e., the round-robin scheduling is not executed).
Let $\rho=\frac{C_{CL}(\zeta=1.005)}{C_{CL}(\zeta\rightarrow\infty)}$, where $C_{CL}(\zeta=1.005)$ and $C_{CL}(\zeta\rightarrow\infty)$ are the average total costs under CL-$\pi c N$ with $\zeta=1.005$ and $\zeta\rightarrow\infty$, respectively. The performance of the algorithms as a function of the error probability for process $1$ is presented in Fig. \ref{fig:fig_gain_explore}. The error probability for process $2$ was set such that $\gamma_1(1)=2\gamma_2(1)$ holds.
It can be seen that setting $\zeta=1.005$ outperforms $\zeta\rightarrow\infty$ as the error probability decreases. This result demonstrates the significance of the round-robin selection rule to guarantee optimality in the asymptotic regime. It should be noted, however, that the loss by removing the round-robin scheduling (i.e., always setting $\zeta\rightarrow\infty$) is small and CL-$\pi c N$ may perform well with $\zeta\rightarrow\infty$ under typical error probabilities.

\begin{figure}[htbp]
\centering \epsfig{file=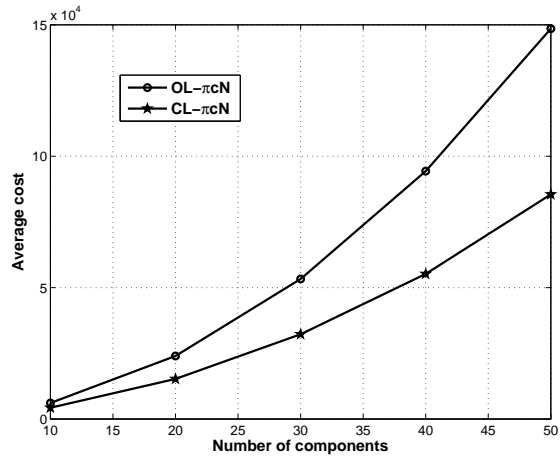,
width=0.45\textwidth}
\caption{The average total cost as a function of the number of components. A case where both strong and weak deviations from the normal state may occur with equal probability.}
\label{fig:fig_composite}
\end{figure}

\begin{figure}[htbp]
\centering \epsfig{file=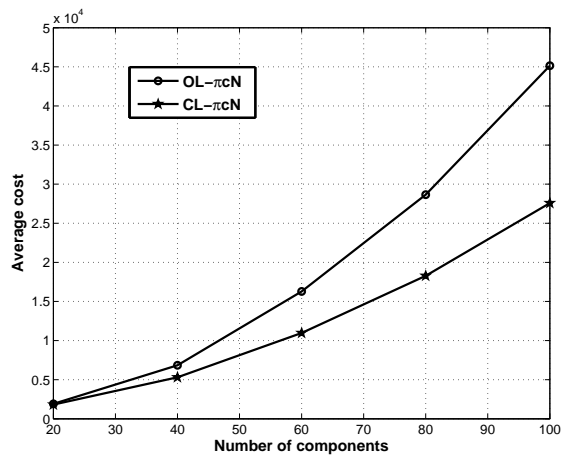,
width=0.45\textwidth}
\caption{The average total cost as a function of the number of components. A case where $M=5$ components are probed at a time.}
\label{fig:fig3}
\end{figure}

\begin{figure}[htbp]
\centering \epsfig{file=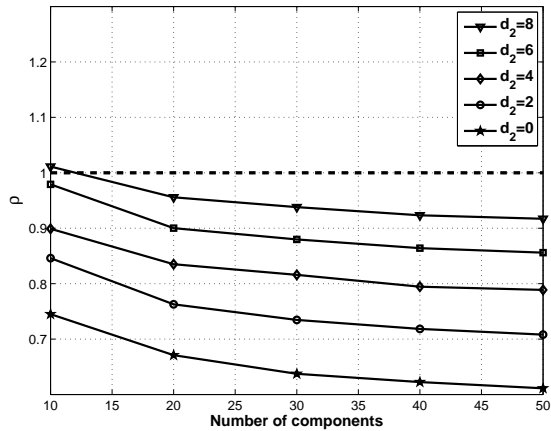,
width=0.45\textwidth}
\caption{The gain $\rho=\frac{C_{CL}}{C_{OL}}$ as a function of the number of components and the delay incurred by switching. Switching to components $1, ..., K/2$ adds delay $d_1=1$ time unit, while switching to components $K/2+1, ..., K$ adds delay $d_2$, which ranges between $0$ to $8$ time units. The CL-$\pi c N$ policy outperforms the OL-$\pi c N$ policy for all $\rho\leq 1$.}
\label{fig:fig2}
\end{figure}

\begin{figure}[htbp]
\centering \epsfig{file=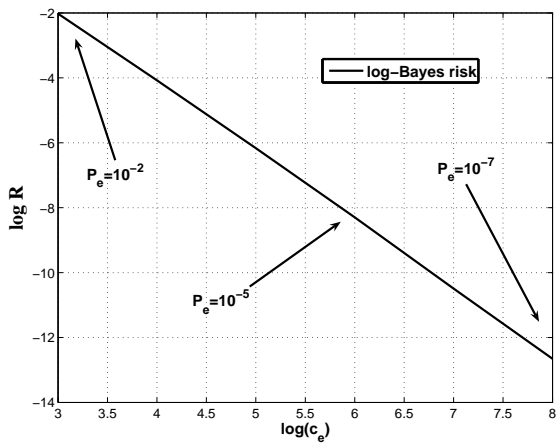,
width=0.45\textwidth}
\caption{The tradeoff curve between the average total cost and the error probabilities (i.e., Bayes risk) as a function of the cost for a wrong declaration.}
\label{fig:Bayes_risk}
\end{figure}

\begin{figure}[htbp]
\centering \epsfig{file=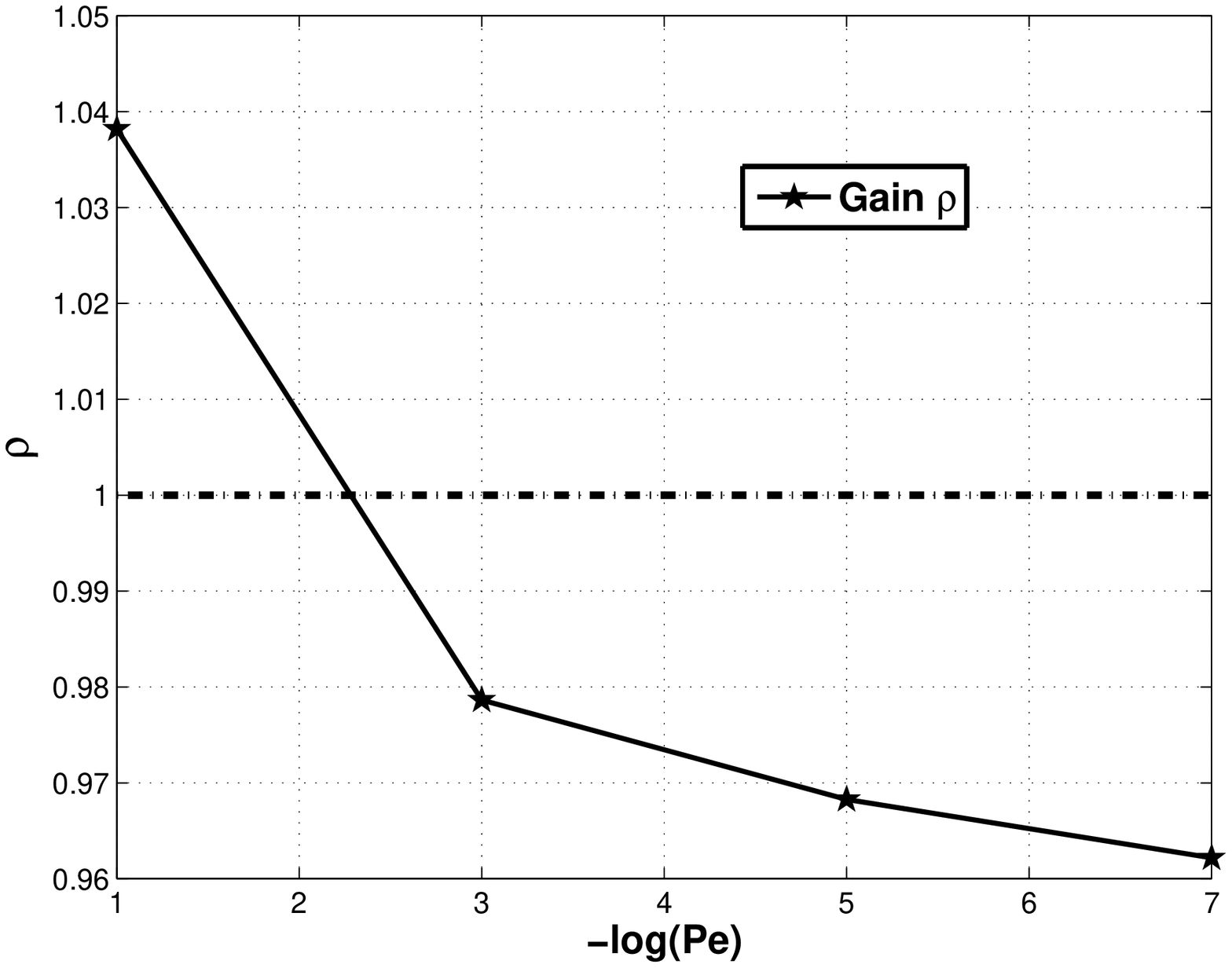,
width=0.45\textwidth}
\caption{The gain $\rho=\frac{C_{CL}(\zeta=1.005)}{C_{CL}(\zeta\rightarrow\infty)}$ as a function of the error probability for process $1$. The CL-$\pi c N$ policy under $\zeta=1.005$ outperforms the CL-$\pi c N$ policy under $\zeta\rightarrow\infty$ for all $\rho\leq 1$.}
\label{fig:fig_gain_explore}
\end{figure}

\section{Conclusion}
\label{sec:conclusion}

The problem of sequential detection of independent anomalous processes among $K$ processes was considered. At each time, only a subset of the processes can be observed, and the observations from each chosen process follow two different distributions, depending on whether the process is normal or abnormal. Each anomalous process incurs a cost per unit time until it is identified. The objective is a sequential search strategy that minimizes the total expected cost incurred by all the processes during the entire detection process, under reliability constraints. Asymptotically optimal closed-loop policies were developed and strong performance in finite regime was demonstrated via simulations as compared to the optimal open-loop policies when the cost incurred by switching across processes is not too high.

\section{Appendix}
\label{app}

In this appendix we prove the asymptotic optimality of the proposed tests as the error constraints approach zero. For purposes of analysis, we assume that the asymptotic expected sample sizes $\E(N_k|H_0), \E(N_{k'}|H_1)$ have the same order for all $k ,k'$. This condition implies that $\log(P_k^{FA})/\log(P_{k'}^{MD})$ is bounded away from zero and infinity for every pair $k, k'$. Throughout the proof, we use the fact that the round-robin selection rule (i.e., second line in (\ref{eq:selection_simple})) observes all the processes according to a predetermined order at times $n=\lceil \zeta^{\ell}\rceil$, for $\ell=1, 2, ...$, where $\zeta$ is a design parameter. We will show that asymptotic optimality holds when $\zeta$ is set sufficiently close to $1$.

Deriving asymptotic optimality is done in two steps. First, we establish the asymptotic lower bound on the total cost that can be achieved by any policy. Second, we show that CL-$\pi c N$ achieves the lower bound in the asymptotic regime. The key in proving the second step is to upper bound the tail of the distribution of some ancillary random times. Specifically, when CL-$\pi c N$ is implemented indefinitely (i.e., CL-$\pi c N$ probes the processes indefinitely according to its selection rule, while the stopping rules and decision rules are disregarded), we can define an event $T_1$ in which for all $n\geq T_1$, the index $\gamma_k(n)$ is a sufficient indication to the process state. The event $T_1$ depends on the future and the true state, and is not a stopping time. The decision maker does not know whether it has arrived. However, we show that $T_1$ is sufficiently small. As a result, we show that when CL-$\pi c N$ is implemented in the asymptotic regime ($P_e^{max}\rightarrow 0$ and thus the detection time approaches infinity), the cost incurred by abnormal processes during the first $T_1$ time units does not affect the asymptotic total expected cost.

\subsection{Proof of Theorem \ref{th:asymptotic_optimality_simple}}
\label{app:asymptotic_optimality_simple}

In this section we prove the asymptotic optimality of CL-$\pi c N$ under the case where the densities are completely known. Note that the SPRT's boundary values (used to test every process) satisfy $B_k=-\log(\alpha_k), A_k=-\log(\beta_k)$ in the asymptotic regime. Let $\E^*(N_k|H_i)$ be the expected sample size for process $k$ under the SPRT. Without loss of generality we assume that $\mathcal{H}_1=\left\{1, 2, ..., K_1\right\}$, $\mathcal{H}_0=\left\{K_1+1, K_1+2, ..., K\right\}$ and\footnote{In cases where processes have the same $c_i/E^*(N_i|H_1)$ , we can arbitrarily order them (by computing their index using a modified cost with an additive small noise $\tilde{c}_k=c_k+\epsilon_k$) without affecting the objective function in the asymptotic regime.}
\beq
\bea{l}
\label{app:ordering}
\displaystyle\frac{c_1}{\E^*(N_1|H_1)}>\frac{c_2}{\E^*(N_2|H_1)}> 
\displaystyle\cdots>\frac{c_{K_1}}{\E^*(N_{K_1}|H_1)}\;.
\ena
\eeq

The proof is mainly based on Lemmas \ref{lemma:lower_simple}, \ref{lemma:asymptotic_simple}. In lemma \ref{lemma:lower_simple}, we establish the asymptotic lower bound on the expected cost that can be achieved by any policy. Then, Lemma \ref{lemma:asymptotic_simple} shows that CL-$\pi c N$ achieves the lower bound in the asymptotic regime. \vspace{0.2cm}

\begin{lemma}
\label{lemma:lower_simple}
Let $\E(C(s))$ be the total expected cost under policy $s$ that satisfies the error constraints in (\ref{eq:opt1}). Then,
\beq
\label{eq:lwer_bound}
\bea{l}
\displaystyle\inf_s \E(C(s)) \geq (1-o(1))\sum_{i=1}^{K_1}c_i\sum_{k=1}^{i}\frac{B_k}{D(f_k^{(1)}||f_k^{(0)})}  \;,
\ena
\eeq
where $o(1)\rightarrow 0$ as $P_e^{max}\rightarrow 0$.
\end{lemma}  \vspace{0.2cm}
\begin{proof}
Note that observing normal processes before declaring the states of abnormal processes can only increase the total expected cost. Hence, for establishing the lower bound on the actual cost we assume that all the abnormal processes are tested before those in a normal state.

Let $\mathbf{y}_k$ be the vector of observations taken from process $k$ and $\mathbf{y}=(\mathbf{y}_1, ..., \mathbf{y}_K)$ be the collection of the observation vectors. Let
\beq
\mathcal{Y}_\epsilon(s)=\left\{\mathbf{y} : N_k>(1-\epsilon)\frac{B_k}{D(f_k^{(1)}||f_k^{(0)})} \;\forall k\right\}
\eeq
be the set of all possible observations collected from the processes with sample sizes satisfying $N_k>(1-\epsilon)\frac{B_k}{D(f_k^{(1)}||f_k^{(0)})}$ for all $k$ under policy $s$. Let $C_{\mathcal{Y}_\epsilon(s)}(\mathbf{y})$ be the total cost incurred by the processes when observations $\mathbf{y}\in\mathcal{Y}_\epsilon(s)$ were taken under policy $s$. \\
Next, we lower bound $C_{\mathcal{Y}_\epsilon(s)}(\mathbf{y})$. We define a modified vector of observations for process $k$, $\tilde{\mathbf{y}}_k$ with length $\tilde{N}_k\triangleq (1-\epsilon)\frac{B_k}{D(f_k^{(1)}||f_k^{(0)})}\leq N_k$ by removing observations $\tilde{N}_k+1, \tilde{N}_k+2, ..., N_k$ for all $k$. The set $\tilde{\mathcal{Y}}_\epsilon(s)$ is defined accordingly as the set of the modified vectors of observations.
Let $C_{\tilde{\mathcal{Y}}_\epsilon(s)}(\tilde{\mathbf{y}})$ be the total cost incurred by the modified vectors of observations, where the selection rule under $s$ skips the time indices that have been removed. As a result, $C_{\tilde{\mathcal{Y}}_\epsilon(s)}(\tilde{\mathbf{y}})\leq C_{\mathcal{Y}_\epsilon(s)}(\mathbf{y})$.

Following the Smith rule \cite{Smith_1956_Various}, minimizing $C_{\tilde{\mathcal{Y}}_\epsilon(s)}(\tilde{\mathbf{y}})$ is done by ordering the processes in decreasing order of $c_k/\tilde{N}_k$. Since $\E^*(N_k|H_1)\rightarrow\tilde{N}_k/(1-\epsilon)$ as $P_e^{max}\rightarrow 0$ \cite{Wald_1947_Sequential}, we have:
\beq
\bea{l}
\label{app:l1_lower_bound}
\displaystyle \inf_s C_{\mathcal{Y}_\epsilon(s)}(\mathbf{y})
\geq(1-\epsilon)\sum_{i=1}^{K_1}c_i\sum_{k=1}^{i}\frac{B_k}{D(f_k^{(1)}||f_k^{(0)})} \vspace{0.3cm} \\ \hspace{5cm}
\mbox{as\;\;} P_e^{max}\rightarrow 0\;.
\ena
\eeq
Finally, we apply \cite[Lemma 2.1]{Tartakovsky_1998_Asymptotic}, where an asymptotic probabilistic lower bound on the sample size achieved by any test (for a single process) that satisfies specific error constraints was established. The lemma was originally stated for a more general case of $M$-ary hypothesis testing and non-i.i.d. observations. It requires a weaker condition on the convergence of a (variation of) the average LLR than the strong law of large numbers. Therefore, it directly applies to the case of binary hypothesis and i.i.d. observations (i.e., the strong law of large numbers implies the convergence of the average LLR to the corresponding KL divergence), considered in this paper. Specifically, applying \cite[Lemma 2.1, Eq. (2.13)]{Tartakovsky_1998_Asymptotic} to our model yields:
\beq
\bea{l}
\displaystyle\inf_s\Pr\left(N_k>\frac{(1-\epsilon)B_k}{D(f_k^{(1)}||f_k^{(0)})}\right)= 1 \;\; \mbox{as} \;\; P_e^{max}\rightarrow 0 \vspace{0.2cm} \\ \hspace{5cm}
\forall k\in\mathcal{H}_1 \;.
\ena
\eeq
Hence, $\Pr\left(\mathbf{y}\in\mathcal{Y}_\epsilon(s)\right)=1$ as $P_e^{max}\rightarrow 0$ for every $\epsilon>0$, which completes the proof.
\end{proof} \vspace{0.2cm}
%
For the next lemmas, we assume that CL-$\pi c N$ is implemented and show that CL-$\pi c N$ achieves the asymptotic lower bound on the expected total cost (\ref{eq:lwer_bound}) as $P_e^{max}\rightarrow 0$.
 \vspace{0.2cm}

\begin{definition}
For every $0<\epsilon<1$, $T_1(\epsilon)$ is defined as the smallest integer such that $\pi_k(n)\geq 1-\epsilon$ for all $k\in\mathcal{H}_1$ and $\pi_k(n)\leq\epsilon$ for all $k\in\mathcal{H}_0$ for all $n\geq T_1(\epsilon)$. \vspace{0.2cm}
\end{definition}
In the following lemma we show that $T_1(\epsilon)$ is sufficiently small.  \vspace{0.2cm}
\begin{lemma}
\label{lemma:T1_eps_simple}
Assume that CL-$\pi c N$ is implemented indefinitely. Then, for every fixed $0<\epsilon<1$ and $\nu>0$, there exists $\delta>0$ such that for all $1<\zeta\leq 1+\delta$ the following holds:
\beq
\label{eq:lemma:T1_eps}
\Pr\left(T_1(\epsilon)>n\right)\leq O(n^{-\nu})\;. \vspace{0.4cm}
\eeq
\end{lemma}
\begin{proof}
Let $d_k\triangleq\frac{1-\pi_k(1)}{\pi_k(1)}$ and
\beq
\bea{l}
\displaystyle M_k^{(1)}\triangleq-\log\left(\frac{\epsilon}{d_k(1-\epsilon)}\right)\;, \vspace{0.3cm}\\
\displaystyle M_k^{(0)}\triangleq-\log\left(\frac{d_k\epsilon}{1-\epsilon}\right)\;.
\ena
\eeq
By rewriting the update formula in (\ref{eq:beleif_ind}), it can be shown that:
\beq
\displaystyle\pi_k(n)=\left(d_k e^{-S_k(n)}+1\right)^{-1}.
\eeq
As a result, $\pi_k(n)\geq 1-\epsilon$ iff $S_k(n)\geq M_k^{(1)}$ and $\pi_k(n)\leq\epsilon$ iff $S_k(n)\leq -M_k^{(0)}$, where $S_k(n)$ is the sum of i.i.d. r.v (i.e., LLR) with mean $\E(\ell_k(n))=D(f_k^{(1)}||f_k^{(0)})>0$ for all $k\in\mathcal{H}_1$ and $\E(\ell_k(n))=-D(f_k^{(0)}||f_k^{(1)})<0$ for all $k\in\mathcal{H}_0$. Since the round-robin selection guarantees that for large $n$, $\log n/(K\log \zeta)$ samples are taken from every process up to time $n$, (\ref{eq:lemma:T1_eps}) follows for an arbitrarily large $\nu$ following the same argument as in \cite{Nitinawarat_2013_Controlled} when $\zeta$ is set sufficiently close to $1$.
\end{proof} \vspace{0.2cm}

\begin{definition}
$T_1$ is defined as the smallest integer such that $\gamma_1(n)>\gamma_2(n)>\cdots>\gamma_{K_1}(n)>\max_{k\in\mathcal{H}_0}\gamma_k$ for all $n\geq T_1$. \vspace{0.2cm}
\end{definition}
Before presenting the next lemma, we provide an intuition for the definition of $T_1$. Assume that no state has been declared by time $T_1$. Then, $T_1$ represents the earliest time where the testing order required to achieve the asymptotic lower bound (i.e., the order: $1, 2, ..., K_1$) is preserved for all $n\geq T_1$. In the following lemma we show that $T_1$ is sufficiently small, such that the cost incurred by abnormal processes during $T_1$ does not affect the asymptotic expected total cost.  \vspace{0.2cm}
\begin{lemma}
\label{lemma:T1}
Assume that CL-$\pi c N$ is implemented indefinitely. Then, for every fixed $\nu>0$, there exists $\delta>0$ such that for all $1<\zeta\leq 1+\delta$ the following holds:
\beq
\label{eq:lemma:T1}
\Pr\left(T_1>n\right)\leq O(n^{-\nu}) \;. \vspace{0.4cm}
\eeq
\end{lemma}
\begin{proof}
Note that Lemma \ref{lemma:T1_eps_simple} holds for any $0<\epsilon<1$ and it is assumed that $\frac{c_1}{\E^*(N_1|H_1)}>\frac{c_2}{\E^*(N_2|H_1)}>\cdots>\frac{c_{K_1}}{\E^*(N_{K_1}|H_1)}$ holds. Since $\gamma_k(n)=\frac{\pi_k(n)c_k}{\pi_k(n)\E^*(N_k|H_1)+(1-\pi_k(n))\E^*(N|H_0)}$ and $\E^*(N_k|H_0), \E^*(N_k|H_1)$ have the same order by assumption, we can choose a sufficiently small $\epsilon>0$ that satisfies the lemma.
\end{proof} \vspace{0.2cm}

In the following lemma we show that the total expected cost under CL-$\pi c N$ approaches the lower bound (\ref{eq:lwer_bound}) as $P_e^{max}\rightarrow 0$.
%
\begin{lemma}
\label{lemma:asymptotic_simple}
Let $\E(C^*)$ be the total expected cost under CL-$\pi c N$. Then,
\beq
\label{eq1:lemma:C*}
\bea{l}
\displaystyle \E(C^*)\sim\sum_{i=1}^{K_1}c_i\sum_{k=1}^{i}\frac{B_k}{D(f_k^{(1)}||f_k^{(0)})} %
\mbox{\;\;as\;\;} P_e^{max}\rightarrow 0 \;. \vspace{0.2cm}
\ena
\eeq\end{lemma} \vspace{0.2cm}
\begin{proof}
Without loss of generality, assume that no state has been declared by time $T_1$ (otherwise, the resulting cost is even smaller than the cost computed below). Thus, for all $n\geq T_1$, CL-$\pi c N$ tests the processes in the following order: $1, 2, ..., K_1$ and then test the normal ones. Let $\bar{c}=\max_k c_k$. Since the total cost incurred up to time $T_1$ is upper bounded by $K\bar{c}T_1$, the total cost $C^*$ under CL-$\pi c N$ is upper bounded by
\beq
\label{app:l4_upper_bound}
\displaystyle C^*\leq K \bar{c} T_1+K_1\bar{c}\sum_{k=1}^{K_1}N_k^s+ \sum_{i=1}^{K_1}c_i\sum_{k=1}^{i}N_k,
\eeq
where $N_k$ is the sample size required to declare the state for process $k$ and $N_k^s$ is the observation sample size due to the round-robin selection rule for process $k$ (i.e., $\E(N_k^s)\leq O(\log B_1)$ in the asymptotic regime since the error probabilities have the same order by assumption.).
%
%
Therefore, applying Lemma \ref{lemma:T1} and using the fact that $\E^*(N_k|H_1)\rightarrow\frac{B_k}{D(f_k^{(1)}||f_k^{(0)})}$ as $P_e^{max}\rightarrow 0$ yields:
\beq
\label{lemma4:eq2}
\bea{l}
\displaystyle \E(C^*)\leq \vspace{0.2cm}\\
\displaystyle O(\log B_1)+(1+o(1))\sum_{i=1}^{K_1}c_i\sum_{k=1}^{i}\frac{B_k}{D(f_k^{(1)}||f_k^{(0)})} \;,
\ena
\eeq
where $o(1)\rightarrow 0$ as $P_e^{max}\rightarrow 0$.\\
Combining (\ref{lemma4:eq2}) and (\ref{eq:lwer_bound}) completes the proof.
\end{proof}

\subsection{Proof of Theorem \ref{th:asymptotic_optimality_composite}}
\label{app:asymptotic_optimality_composite}

In this section we prove the asymptotic optimality of the proposed policy when the densities have unknown parameters. For purposes of analysis we consider the model in \cite{Chernoff_1959_Sequential}, where $\theta_k$ can take only a finite number of values. Throughout the proof we omit steps that use similar arguments as in the proof under the case of completely known densities.

Using a similar argument as in Lemma \ref{lemma:lower_simple}, it can be shown that
\beq
\label{eq:lwer_bound_composite}
\bea{l}
\displaystyle\inf_s
\E(C(s))\sim\sum_{i=1}^{K_1}c_i\sum_{k=1}^{i}\frac{B_k^{(0)}}
            {D_k^*(\theta_k||\Theta_k^{(0)})} 
\;\;\mbox{as\;\;} P_e^{max}\rightarrow 0 \;.
\ena
\eeq
Next, we show that CL-$\pi c N$ achieves this bound.  \vspace{0.2cm}

\begin{definition}
$T_{ML}$ is defined as the smallest integer such that $\hat{\theta}_k(n)=\theta_k$ for all $k$ for all $n\geq T_{ML}$. \vspace{0.2cm}
\end{definition}
In the following lemma we show that $T_{ML}$ is sufficiently small.  \vspace{0.2cm}
\begin{lemma}
\label{lemma6:T_ML}
Assume that CL-$\pi c N$ is implemented indefinitely. Then, for every fixed $\nu>0$, there exists $\delta>0$ such that for all $1<\zeta\leq 1+\delta$ the following holds:
\beq
\label{eq:lemma6:T_ML}
\Pr\left(T_{ML}>n\right)\leq O(n^{-\nu}) \;. \vspace{0.4cm}
\eeq
\end{lemma}
\begin{proof}
Note that when $K=1$ (i.e., all the observations are taken from a single process), $\Pr\left(T_{ML}>n\right)$ decays exponentially with $n$ following the same argument as in \cite{Chernoff_1959_Sequential}. Furthermore, for large $n$, at least $\log n/(K\log \zeta)$ samples are taken from every process by time $n$. Thus, (\ref{eq:lemma6:T_ML}) follows when $\zeta$ is set sufficiently close to $1$.
\end{proof}  \vspace{0.2cm}

\begin{definition}
For every $0<\epsilon<1$, $T_1(\epsilon)$ is defined as the smallest integer such that $\hat{\pi}_k(n)\geq 1-\epsilon$ for all $k\in\mathcal{H}_1$ and $\hat{\pi}_k(n)\leq\epsilon$ for all $k\in\mathcal{H}_0$ for all $n\geq T_1(\epsilon)$. \vspace{0.2cm}
\end{definition}
In the following lemma we show that $T_1(\epsilon)$ is sufficiently small.  \vspace{0.2cm}
\begin{lemma}
\label{lemma:T1_eps}
Assume that CL-$\pi c N$ is implemented indefinitely. Then, for every fixed $0<\epsilon<1$ and $\nu>0$, there exists $\delta>0$ such that for all $1<\zeta\leq 1+\delta$ the following holds:
\beq
\label{eq1:lemma:T1_eps}
\Pr\left(T_1(\epsilon)>n\right)\leq O(n^{-\nu}) \;. \vspace{0.4cm}
\eeq
\end{lemma}
\begin{proof}
Note that:
\beq
\bea{l}
\Pr\left(T_1(\epsilon)>n\right)\vspace{0.2cm}\\\hspace{0.5cm}
\leq\Pr\left(T_1(\epsilon)>n,T_{ML}\leq n\right)+\Pr\left(T_{ML}> n\right)\;.
\ena
\eeq
The term $\Pr\left(T_{ML}> n\right)$ decays polynomially with $n$ by applying Lemma \ref{lemma6:T_ML}. Thus, it suffices to show that $\Pr\left(T_1(\epsilon)>n,T_{ML}\leq n\right)$ decays polynomially with $n$.

Let $d_k\triangleq\frac{1-\pi_k(1)}{\pi_k(1)}$ and
\beq
\bea{l}
\displaystyle M_k^{(1)}\triangleq-\log\left(\frac{\epsilon}{d_k(1-\epsilon)}\right)\;, \vspace{0.2cm}\\
\displaystyle M_k^{(0)}\triangleq-\log\left(\frac{d_k\epsilon}{1-\epsilon}\right)\;.
\ena
\eeq
By rewriting the update formula in (\ref{eq:beleif_ind}), it can be shown that:
\beq
\label{eq2:lemma:T1_eps}
\displaystyle\hat{\pi}_k(n)=\left(d_k e^{-S_k^{(1),GLR}(n)}+1\right)^{-1},
\eeq
for all $k\in\mathcal{H}_1$ for all $n\geq T_{ML}$,\\
and
\beq
\label{eq2:lemma:T1_eps_H_0}
\displaystyle\hat{\pi}_k(n)=\left(d_k e^{S_k^{(0),GLR}(n)}+1\right)^{-1},
\eeq
for all $k\in\mathcal{H}_0$ for all $n\geq T_{ML}$.\\
As a result, $\hat{\pi}_k(n)\geq 1-\epsilon$ iff $S_k^{(1),GLR}(n)\geq M_k^{(1)}$ for all $k\in\mathcal{H}_1$ and $\hat{\pi}_k(n)\leq\epsilon$ iff $S_k^{(0),GLR}(n)\geq M_k^{(0)}$ for all $k\in\mathcal{H}_0$ for all $n\geq T_{ML}$. Thus, it suffices to show that $\Pr(S_k^{(1),GLR}(n)\leq M_k^{(1)}|n\geq T_{ML})$ for all $k\in\mathcal{H}_1$ and $\Pr(S_k^{(0),GLR}(n)\leq M_k^{(0)}|n\geq T_{ML})$ for all $k\in\mathcal{H}_0$
decay polynomially with $n$.
Note that when $T_{ML}\leq n$ occurs, $S_k^{(0),GLR}(n)$ for all $k\in\mathcal{H}_1$ and $S_k^{(1),GLR}(n)$ for all $k\in\mathcal{H}_0$ are sums of i.i.d. r.v. with positive KL divergence (since $\hat{\theta}_k(n)=\theta_k$ for all $n\geq T_{ML}$). Since at least $\log n/(K\log \zeta)$ samples are taken from every process by time $n$, the lemma follows.
\end{proof}  \vspace{0.2cm}

The rest of the proof follows with minor modifications to the proof under the case of completely known densities.

\subsection{Proof of Theorem \ref{th:asymptotic_optimality_multi}}
\label{app:asymptotic_optimality_multi}

In this appendix we prove the asymptotic optimality of CL-$\pi c N$ under multi-process probing when $c\triangleq c_1=c_2=\cdots=c_K$. Throughout the proof we omit steps that use similar arguments as in the proof under single-process probing. We also use similar notations as in App. \ref{app:asymptotic_optimality_simple}.

First, we establish the asymptotic lower bound on the expected cost that can be achieved by any policy. Using the same notations as in the proof of Lemma \ref{lemma:lower_simple}, we aim to lower-bound $C_{\mathcal{Y}_\epsilon(s)}(\mathbf{y})$ using the definition of $C_{\tilde{\mathcal{Y}}_\epsilon(s)}(\tilde{\mathbf{y}})$.
Recall that $C_{\tilde{\mathcal{Y}}_\epsilon(s)}(\tilde{\mathbf{y}})$ is the total cost incurred by the modified vectors of observations with a fixed sample size.

Next, we apply \cite[Theorem 5.4.2]{Pinedo_2012_Scheduling} to minimize $C_{\tilde{\mathcal{Y}}_\epsilon(s)}(\tilde{\mathbf{y}})$. In \cite{Pinedo_2012_Scheduling}, the problem of ordering jobs with fixed processing times over $M$ parallel machines was considered. It was shown that scheduling the jobs in decreasing order of $1/\tilde{N}_k$, where $\tilde{N}_k$ is the processing time for job $k$, minimizes the sum completion times of the jobs. When applying \cite[Theorem 5.4.2]{Pinedo_2012_Scheduling} to our case, the sum completion times for the modified observation vectors is $\frac{1}{c}C_{\tilde{\mathcal{Y}}_\epsilon(s)}(\tilde{\mathbf{y}})$ when all the abnormal processes incur the same cost $c$ per unit time. Since $c=c_1=\cdots=c_K$ by assumption (and in particular $c=c_1=\cdots=c_{K_1}$ for any realization of the true system state), we can apply \cite[Theorem 5.4.2]{Pinedo_2012_Scheduling}. As a result, minimizing $C_{\tilde{\mathcal{Y}}_\epsilon(s)}(\tilde{\mathbf{y}})$ is done by ordering the processes in decreasing order of $1/\tilde{N}_k$. Let
\beq
\displaystyle \tilde{c}_k=
\begin{cases}
c \;,\;\mbox{if\;} k\in \mathcal{H}_1\;, \vspace{0.2cm}\\
0 \;,\;\mbox{otherwise} \;.
\end{cases}
\eeq
Note that minimizing $C_{\tilde{\mathcal{Y}}_\epsilon(s)}(\tilde{\mathbf{y}})$ by ordering the modified observation vectors in decreasing order of $1/\tilde{N}_k$ implies that at each given time the $M$ vectors with the smallest sample sizes among the remaining vectors contribute to the total cost. As a result, Similar to (\ref{app:l1_lower_bound}), for any $\epsilon>0$, we can lower bound the actual cost by the cost achieved by minimizing $C_{\tilde{\mathcal{Y}}_\epsilon(s)}(\tilde{\mathbf{y}})$:
\beq
\bea{l}
\displaystyle \inf_s C_{\mathcal{Y}_\epsilon(s)}(\mathbf{y})
\geq(1-\epsilon)\sum_{m=1}^{M}\sum_{i=1}^{\lceil K_1/M\rceil}\tilde{c}_{m+(i-1)M} \times\vspace{0.2cm} \\ \hspace{2.5cm}
\displaystyle\sum_{k=1}^{i}\frac{B_{m+(k-1)M}}{D(f_{m+(k-1)M}^{(1)}||f_{m+(k-1)M}^{(0)})} \vspace{0.3cm} \\ \hspace{5cm}
\mbox{as\;\;} P_e^{max}\rightarrow 0\;,
\ena
\eeq
Hence, following the same argument as in Lemma \ref{lemma:lower_simple}, we obtain:
\beq
\label{eq:lower_bound_multi}
\bea{l}
\displaystyle\inf_s \E(C(s)) \geq (1-o(1))\sum_{m=1}^{M}\sum_{i=1}^{\lceil K_1/M\rceil}\tilde{c}_{m+(i-1)M} \times\vspace{0.2cm} \\ \hspace{2.5cm}
\displaystyle\sum_{k=1}^{i}\frac{B_{m+(k-1)M}}{D(f_{m+(k-1)M}^{(1)}||f_{m+(k-1)M}^{(0)})}  \;,
\ena
\eeq
where $o(1)\rightarrow 0$ as $P_e^{max}\rightarrow 0$.

Next, we show that CL-$\pi c N$ achieves the lower bound (\ref{eq:lower_bound_multi}) in the asymptotic regime. Following the definition of $T_1$, for all $n\geq T_1$, CL-$\pi c N$ tests the processes in the desired order required to obtain the lower bound as specified in (\ref{eq:lower_bound_multi}). Note that by applying Lemma \ref{lemma:T1}, we can set $\zeta>1$ sufficiently close to $1$, such that $\Pr\left(T_1>n\right)\leq O(n^{-\nu})$ for an arbitrarily large $\nu>0$. Therefore, similar to (\ref{app:l4_upper_bound}), (\ref{lemma4:eq2}), we have:
\beq
\label{app:pr3_upper_bound}
\bea{l}
\displaystyle \E(C^*)\leq(1+o(1))\sum_{m=1}^{M}\sum_{i=1}^{\lceil K_1/M\rceil}\tilde{c}_{m+(i-1)M} \times\vspace{0.2cm} \\ \hspace{1cm}
\displaystyle\sum_{k=1}^{i}\frac{B_{m+(k-1)M}}{D(f_{m+(k-1)M}^{(1)}||f_{m+(k-1)M}^{(0)})}+O(\log B_1)  \;,
\ena
\eeq
where $o(1)\rightarrow 0$ as $P_e^{max}\rightarrow 0$.\\
Combining (\ref{eq:lower_bound_multi}) and (\ref{app:pr3_upper_bound}) completes the proof.

\bibliographystyle{ieeetr}

%
%
\end{document}